\documentclass[11pt,a4paper,english]{article}

\title{Massless QFT and the  Newton-Wigner Operator}
\author{A. Much\\ \footnotesize{Instituto de Ciencias Nucleares, UNAM, M\'exico D.F. 04510, M\'exico},
	\\ \footnotesize{Faculty of Mathematics, University of Vienna, 1090 Vienna, Austria}}

\usepackage{float}            % bietet Option [H] fÃƒÂ¼r bombenfestes Verankern
\usepackage[T1]{fontenc}   
\usepackage[cmex10]{amsmath}  % fÃƒÂ¼r erw. Formeloptionen, Option [] zur Vermeidung von Type3-Fonts
\usepackage{amstext}          % fÃƒÂ¼r Klartext via \text{} in Formeln
\usepackage{mathrsfs}
\usepackage{amsfonts}         % fÃƒÂ¼r komplexere Formeln (Mengensymbole ...)
\usepackage{amssymb}          % fÃƒÂ¼r komplexere Formeln (Mengensymbole ...)
\usepackage{amsthm}    
\usepackage{setspace}
\usepackage{bm}               % bold math, fÃƒÂ¼r \bm{}
\usepackage{enumerate}        % verbessert AufzÃƒÂ¤hlungen
\usepackage[bottom]{footmisc} % Fussnoten am Seitenende
\usepackage{array}            % fÃƒÂ¼r Tabellen: bindet tabular-Umgebung ein
\usepackage{textcomp}         % fÃƒÂ¼r \textdegree , \textcelsius , macht aber manchmal Probleme (!!)
\usepackage{pdfpages}         % fÃƒÂ¼r die Einbindung kompletter pdf-*Seiten*
\usepackage{parskip}          % zw. AbsÃƒÂ¤tzen: eine knappe Leerzeile statt hÃƒÂ¤ngender EinzÃƒÂ¼ge
\usepackage[right]{eurosym}   % Eurosymbol
\usepackage{xcolor}           % farbiger Text
\usepackage[square,numbers]{natbib}	  % FÃƒÂ¼r \setlength{\bibsep}{3mm}; square macht eckige

\usepackage{amsmath}  
\usepackage[colorlinks,pdfpagelabels,pdfstartview = FitH,bookmarksopen = true,bookmarksnumbered =
true,linkcolor = black,plainpages = false,hypertexnames = false,citecolor = black]{hyperref}

\newtheorem{theorem}{\textsc{Theorem}}[section]
\newtheorem{lemma}{\textsc{Lemma}}[section]
\newtheorem{proposition}{\textsc{Proposition}}[section]

\theoremstyle{definition}
\newtheorem{definition}{\textsc{Definition}}[section]

\newtheorem{convention}{Conventions}[section]

\theoremstyle{remark}
\newtheorem{remark}{Remark}[section]

\newcommand{\R}{\mathbb{R}}

\usepackage{fancyhdr}
\numberwithin{equation}{section} 
\begin{document}
	\maketitle	\abstract{ In this work, the second-quantized version of the spatial-coordinate operator, known as the  Newton-Wigner-Pryce operator, is explicitly given  w.r.t. the massless scalar field. Moreover,   transformations of the conformal group are calculated on eigenfunctions of this operator in order to investigate the covariance group w.r.t. probability amplitudes of localizing particles.} 
	%%%%%%%%%%%%%%%%%%%%%%
	% Inhaltsverzeichnis %
	\tableofcontents
	%%%%%%%%%%%%%%%%%%%%%%

\section{Introduction} 
Localization is  from a mathematical and physical point of view considered to be an unresolved issue in a relativistic context, see for example  \cite{NW49}, \cite{F1},  \cite{H2}, \cite{H1}  and references therein. The concept is well-defined from a non-relativistic quantum-mechanical (QM) point of view. However, by combining   concepts of QM and relativity, from which the frame-work of quantum field theory (QFT) emerged, the localization issue was never completely resolved. This is not merely an issue of construction but points to a more fundamental question about  the existence of localizable particles.
\\\\
The authors in \cite{NW49} tackled the problem by demanding certain requirements from the states of a coordinate operator, which are used to localize elementary systems. Independently, among other operators the author in \cite{PR48} gave as well a definition of an appropriate coordinate operator, which is equivalent to the one obtained by \cite{NW49}. Although the operator is unique for massive particles with arbitrary spin and for every massless system with spin $0,\,1/2$, it has two major physical setbacks. First,   eigenstates of the Newton-Wigner-Pryce (NWP) operator  propagate superluminally and moreover the states are not covariant w.r.t. general Lorentz-transformations. Hence, there are two different opinions on this subject; either such localizable particles do not exist in a relativistic context or one can still make   physical sense of the Newton-Wigner-Pryce   operator and its respective states. 
\\\\
In this work, we take the second point of view and give, in our opinion, convincing arguments why the NWP sates are of physical interest. In particular, the importance of these states relies on the fact that they can be used to calculate the probability amplitude  of finding a particle at a certain spatial-position at time $t$. Moreover, the second quantization of this operator allows us  to calculate the probability amplitude  of finding $k$-particles at   certain spatial-positions at time $t$.
\\\\
In order to  have a richer example of the NWP-space in a quantum field theoretical context (than the massive scalar field \cite{Muc7}), we work with the free \textbf{massless} scalar field. In particular, we show that the localizability of  particles, in the massless case, is    not a lost cause. In general, the covariance group in the massless case is more than solely the Poincar\'e group,  it is given by the conformal group. Hence,   we investigate the covariant transformation property of NWP-states under conformal transformations. The importance therein has a physical reasoning. It is first of all done in order to prove that although the NWP-states are non-covariant w.r.t. to the whole conformal group there exist transformations   of a subgroup of the conformal group that preserve their covariant character.  Furthermore, given two observers  measuring, for a given covariant state, the probability of finding $k$-particles at $k$-spatial-positions. If these two observers are in two relatively to each other transformed  frames, for example one is rotated, space-time translated or dilated to the other  then they can agree on the observation of the probability of spatial positions, by using the covariant transformation rules.
\\\\
Besides  calculations of  position-probability amplitudes of particles, the NWP-operator has an additional physical relevance. In this work we prove that by solely using the coordinate  and relativistic momentum operator, one is able to define all generators of the conformal group. Hence, in some sense, the NWP-operator can be seen as essential for implementing relativistic symmetries. 
\\\\
Another motivation besides understanding localization   to translate the conformal group into the coordinate space, is non-commutative (NC) geometry applied to quantum field theory (QFT). In particular the geometry that is quantized in NCQFT is the geometry of the NWP-space. Hence, in order to understand the outcome and the action of a deformation of  QFT using, for example (see   \cite{MUc1})  the conformal group,  we need a better understanding of the commutative case.\\\\
The paper is organized as follows; Section two gives a quick introduction to the tools needed in this paper, where we define the Fock space of the free massless scalar field, the  Newton-Wigner-Pryce operator and  the algebra of the Fourier-transformed ladder operators.   In section four, the conformal algebra  for a massless scalar field is transformed into the NWP-space.
Section five investigates the covariance of the NWP-space under the Conformal group.

 \begin{convention}
 	We use $d=n+1$, for $n\in\mathbb{N}$ and the Greek letters are split into  $\mu, \,\nu=0,\dots,n$. Moreover, we use Latin letters for the spatial components which run from $1,\dots,n$ and we choose the following convention for the Minkowski scalar product of $d $-dimensional vectors, $a\cdot b=a_0b^0+a_kb^k=a_0b^0- \vec{a}\cdot\vec{b}$. Furthermore, we use the common symbol $\mathscr{S}(\mathbb{R}^n)$ for the Schwartz-space.
 \end{convention} 
\section{Preliminaries}
\subsection{Bosonic Fock-space and Fourier-transformation}
We briefly define  the  Bosonic Fock 
space for a free scalar field 
 $\phi$ with mass $m=0$ on the $n+1$-dimensional Minkowski spacetime. In particular a  particle with momentum $\mathbf{p} \in \mathbb{R}^n$ has 
 the energy $\omega_{\mathbf{p} }$ given by $\omega_{\mathbf{p} }=|\mathbf{p}|$. Another quantity   needed for the definition of the Fock-space is the Lorentz-invariant measure: $d^n\mu(\mathbf{p} )=d^n\mathbf{p}/( {2\omega_{\mathbf{p}}})$.
 \begin{definition} The Bosonic Fock space $\mathscr{H}^{+}$ is defined 
 	as in \cite{Fr,S}:
 	\begin{equation*}
 	\mathscr{H}^{+}=\bigoplus_{k=0}^{\infty}\mathscr{H}_{k}^{+}
 	\end{equation*}
 	where the $k$ particle subspaces are given as
 	\begin{align*}
 	\mathscr{H}_{k}^{+}&=\{\Psi_{k}: \partial V_{+} \times  \dots \times 
 	\partial V_{+} \rightarrow \mathbb{C}\quad \mathrm{symmetric}
 	|\\ &\left\Vert  \Psi_k \right\Vert^2 =\int 
 	d^n\mu(\mathbf{p_1})\dots\int d^n\mu(\mathbf{p_k})
 	|\Psi_{k}(\mathbf{p_1},\dots,\mathbf{p_k})|^2<\infty\},
 	\end{align*}
 	with
 	\begin{equation*}
 	\partial V_{+}:=\{p\in \mathbb{R}^{n+1}|p^2=0,p_0>0\}.
 	\end{equation*}
 \end{definition}
 The particle annihilation and creation operators can be defined by their 
 action on $k$-particle wave functions
 \begin{align*}
 (a_c(f)\Psi)_k(\mathbf{p_1},\dots,\mathbf{p_k})&=\sqrt{k+1}\int 
 d^n\mu(\mathbf{p})\overline{f(\mathbf{p})}
 \Psi_{k+1}(\mathbf{p},\mathbf{p_1},\dots,\mathbf{p_k})\\
 (a_c(f)^{*}\Psi)_k( \mathbf{p },\mathbf{p_1},\dots,\mathbf{p_k})&= \left\{
 \begin{array} {cc}
 0, \qquad &k=0 \\ \frac{1}{\sqrt{k}}\sum\limits_{l=1}^{k} f(\mathbf{p_l})
 \Psi_{k-1}(\mathbf{p_1},\dots,\mathbf{p_{l-1}},\mathbf{p_{l+1}},\dots,\mathbf{p_k}),\quad 
 &k>0
 \end{array} \right.
 \end{align*}
 with $f \in \mathscr{H}_{1} $ and $\Psi_k \in \mathscr{H}_{k}^{+}$ . The commutation relations of the smeared annihilation and creation operators $a_c(f), 
 a_c(f)^{*}$ follow from their action on the functions $\Psi$ and are given as follows
 \begin{align*}
 [a_c(f), a_c(g)^{*}]=\langle f,g\rangle=\int d^n\mu(\mathbf{p}) \overline{f(\mathbf{p})} 
 g(\mathbf{p}), \qquad
 [a_c(f), a_c(g)]=0=[a_c^{*}(f), a_c^{*}(g)].
 \end{align*}
The ladder operators  with sharp momentum are 
 introduced as operator valued distributions in the following
 \begin{align*}
 a_c(f)=\int d^n\mu(\mathbf{p}) \overline{f(\mathbf{p})}a_c(\mathbf{p}), 
 \qquad a_c(f)^{*}=\int d^n\mu(\mathbf{p}) {f(\mathbf{p})}a_c^{*}(\mathbf{p}),
 \end{align*}
 where the  ladder operators with sharp 
 momentum satisfy the well-known  commutator relations
 \begin{align*}
 [a_c(\mathbf{p}), a_c(\mathbf{q})^{*}]=2\omega_{\mathbf{p} 
 }\delta^n(\mathbf{p}-\mathbf{q}), \qquad
 [a_c(\mathbf{p}), a_c(\mathbf{q})]=0=[a_c^{*}(\mathbf{p}), a_c^{*}(\mathbf{q})].
 \end{align*}
 In the following sections, we use the non-covariant normalization given by, 
 
 \begin{align}
 {a}(\mathbf{p})= \frac{{a}_c(\mathbf{p})}{\sqrt{2\omega_{\mathbf{p} 
 	}}},\qquad {a}^{*}(\mathbf{p})= \frac{{a}^{*}_c(\mathbf{p})}{\sqrt{2\omega_{\mathbf{p} 
 }}}.
 \end{align}
 Next, we define the Fourier-transformation of the ladder operators that were given in the former expressions. These transformations are of physical importance in regards to the following sections.  
  \begin{definition}\label{sec2.2}\textbf{Fourier-transformation}\\\\
  	In order to change from momentum space to the NWP-space we use   explicit expressions for the Fourier-transformed creation and annihilation operators which are given by,
  	\begin{align*}
  	{a}(\mathbf{p})=(2\pi)^{-n/2} \int
  	d^n \mathbf{x}\, e^{ip_kx^k} \tilde{a}(\mathbf{x}),\qquad {a}^{*}(\mathbf{p})=(2\pi)^{-n/2} \int
  	d^n \mathbf{x}\, e^{-ip_kx^k} \tilde{a}^*(\mathbf{x}).
  	\end{align*}
  	The commutation relation between the ladder operators in momentum space gives us directly  the relations for   ladder operators of the NWP space, 
  	\begin{align*}
  	\delta^n(\mathbf{p}-\mathbf{q})&=[{a}(\mathbf{p}),{a}^*(\mathbf{q})]=(2\pi)^{-n} 
  	\iint
  	d^n \mathbf{x}\,d^n \mathbf{y}\, e^{ip_kx^k}\, e^{-iq_ky^k}\underbrace{ [\tilde{a}(\mathbf{x}),\tilde{a}^*(\mathbf{y})]}_{=\delta^n(\mathbf{x}-\mathbf{y})}.
  	\end{align*} 
  Inverse Fourier-transformations of momentum space operators to the NWP-space ladder operators  are given by,
  	\begin{align}\label{inf}
  	\tilde{a}(\mathbf{x})=(2\pi)^{-n/2} \int
  	d^n \mathbf{p}\, e^{-ip_kx^k}{a}(\mathbf{p}),\qquad 
  	\tilde{a}^*(\mathbf{x})=(2\pi)^{-n/2} \int
  	d^n \mathbf{p}\, e^{ ip_kx^k}{a}^*(\mathbf{p}).
  	\end{align}

  \end{definition} 
 
\subsection{Massless NWP-operator}
The   (spatial)-QFT-position operator for a massive free scalar  field is defined by the  Newton-Wigner-Pryce  operator, \cite{PR48} and \cite{NW49}. On a one particle wave-function it acts as follows, \cite[Chapter 3c, Equation 35]{Sch}
 \begin{equation}\label{NWP}
 (X_{j} \varphi)(\mathbf{p})=-i \left( \frac{p_j}{2\omega_{m,\mathbf{p}}^2}
 +   \frac{\partial}{\partial p^j } 
 \right)\varphi(\mathbf{p}),
 \end{equation}
  where $\omega_{m,\mathbf{p}}$ is the relativistic-energy for a   particle with mass $m$. 
 In order to prove that the NWP-operator is in the massless case equivalently represented as in the massive case, we can proceed in different ways. First, we can take Equation (\ref{NWP}) and perform the massless limit. Since, the energy $\omega_{m,\mathbf{p}}$ is the only term that is affected by the limit, the  massless  NWP-operator is given by exchanging in Formula (\ref{NWP})  the massive energy with $\omega_{\mathbf{p}}$. Therefore, the coordinate operator for a massless QFT has the same form as in the massive case. Another possible path to proceed is to define the coordinate operator as unitary equivalent to the  second-quantized spatial-component of the relativistic momentum operator. In terms of the  Fock space operators the relativistic momentum operator is given  as,
 \begin{align}\label{mop}
 P_{\mu}&= \int d^n \mathbf{p}\, p_{ \mu}\, {a}^*(\mathbf{p}) {a}(\mathbf{p})
 .
 \end{align} 
 The unitary equivalence of the coordinate   to the spatial momentum operator is given by the unitary map represented by the Fourier-transformation. In particular, it was proven that the NWP-operator can be represented as a self-adjoint operator on the  domain,  $\bigotimes_{i=1}^k \mathscr{S}(\mathbb{R}^n)$, with  $\mathscr{S}(\mathbb{R}^n)$ denoting the Schwartz space, for details see \cite{Muc2} and \cite{Muc3}. The explicit Fock-space representation of the spatial-coordinate operator  is given by,
 \begin{align*}
 X_j&=-i\int d^n \mathbf{p}\,  {a}^*(\textbf{p}) \frac{\partial}{\partial p^j}  {a}(\textbf{p}).
 \end{align*} 
 Since the spatial-momentum operator has the same form in the massless case as in the massive one, the NWP-operator, defined through the unitary equivalence, takes the same form as in the massive case. This is stems from the fact that   ladder operators do not explicitly depend on the mass.  The commutation relations between the spatial-momentum operator and the NWP-operator are given by, see \cite{SS} or \cite{Muc3},
 \begin{align}\label{ccr}
 [X_j,P_k]&=-i\eta_{jk}N,
 \end{align} 
 where $N$ is the particle-number operator represented in Fock-space as
 \begin{align}\label{pn}	N=  \int
 d^n \mathbf{p}\,     {a}^*(\mathbf{p}) {a}(\mathbf{p}).
 \end{align}
 Next, we want to give the physical interpretation and relevance of this operator.   The eigenfunctions of the Newton-Wigner-Pryce  operator,   simultaneously representing the localized wave functions at time $x_{0}=0$, are given by, \cite[Chapter 3, Equation 38]{Sch}
\begin{align*} 
\Psi_{\mathbf{x},0}(\mathbf{p})=(2\pi)^{-n/2}\,e^{-i\mathbf{p} \cdot \mathbf{x}}\,(2\omega_{p})^{1/2}.
\end{align*}  
 Let a free massless scalar field  be in a state $\Phi(\mathbf{p})$ at time $t=0$. Then, the probability amplitude of finding a particle at the position $\mathbf{x}$ is given by, \cite[Chapter 3, Equation 44]{Sch}
\begin{align*} 
\langle \overline{\Psi}_{\mathbf{x}}, \Phi\rangle=\int d^n\mu(\mathbf{p})\, \Psi_{\mathbf{x},0}(\mathbf{p}) \Phi(\mathbf{p}).
\end{align*}  
In order to extend  this quantity to $k$-particles  to calculate  the probability amplitude for finding   $k$-particle at positions $\mathbf{x}_{1} \cdots \mathbf{x}_{k}$ at time $x_{0}={x}_{10}= \cdots={x}_{k0}$, the following operator is introduced,  \cite[Chapter 7, Equation 99]{Sch},
 \begin{align}\label{copes}
 \phi_1( {x}) =\int d^n\mu(\mathbf{p}) \overline{\Psi_{\mathbf{x},x_{0}}(\mathbf{p})} a_c(\mathbf{p}) =(2\pi)^{-n/2}\,
 \int d^n \mathbf{p} \, e^{-i {p} \cdot {x}}\,  a (\mathbf{p}),
 \end{align}
 and we apply $k$-times $\phi_1$  on $\vert\Phi^1  \rangle$ and the vacuum as  follows,
 \begin{align*} 
 \Phi^1( x_1 ,\cdots, x_k ) = (k!)^{-1/2} \langle 0\vert \phi_1(x_{1})\cdots\phi_1(x_{k})
 \vert\Phi^1  \rangle.
 \end{align*}
  It is important to point out that the operator $\phi_1(x)$ at time $x^{0}=0$ is nothing else than the Fourier-transformed annihilation operator, see Equation (\ref{inf}). Moreover, 	the  Newton-Wigner-Pryce operator has the following coordinate space representation (for proof see \cite[Lemma 3.1]{Muc7}),
  \begin{align*}
  X_j&= \int d^n \mathbf{x}\,  x_j \, \tilde{a}^*(\textbf{x}) \tilde{a}(\textbf{x})
  .
  \end{align*} 
\subsection{Constantly Used Integrals}
To make this work self-contained, we give in this section a general formula for certain Fourier-transformed functions. In particular, these formulas can be found in \cite[Chapter III, Section 2.6-2.8]{GS1},
\begin{align}\label{f1}
\tilde{P}^{\lambda}&= \int\,d^n\mathbf{p}\, \vert\vec{p}\vert^{2\lambda} \exp({-i\,p_k z^k})\\\nonumber
&= 2^{2\lambda+n}\pi^{\frac{1}{2}n}\,\frac{\Gamma(\lambda+\frac{1}{2}n)}{\Gamma(-\lambda)}\left(z_1^2+\cdots +z_n^2\right)^{-\lambda-\frac{1}{2}n},
\end{align}
where $\Gamma$ denotes the Gamma-function.
 
\section{Conformal Group in the  NWP-space}
In this section   we change the basis of the conformal group, which is given in terms of Fock-space operators in the momentum space,   to the coordinate space. The  base changes are mathematically well-defined  on the domain $\bigotimes_{i=1}^k \mathscr{S}(\mathbb{R}^n)$. This conclusion follows from the fact that we use  the  Fourier-transformation for the base change and this transformation acts as a linear automorphism on the Schwartz-space. Hence, we conclude that the base change is well-defined.
\\\\However, an exception to the rule is given by the Lorentz generators generating boosts. They  do not map the Schwartz-space to itself since   the coefficient functions of the differential operators are
not differentiable at $p=0$. Hence, if one intends to calculate the adjoint action w.r.t.  the Lorentz-boosts one has to choose the space of analytic functions on $S^2 \times \R$, more specifically of analytic sections, which is the light-cone without the tip, (see \cite[ Chapter 8.6]{ND}).\footnote[1]{I am indebted to Prof. N. Dragon for this remark.}
\\\\
From a  quantum field theoretical point of view,  the conformal group is of great interest. 
In particular, the group is defined as the set of all conformal transformations. The definition of a conformal transformation is  an 
invertible mapping $\mathbf{x}'\rightarrow \mathbf{x}$, leaving a
$d$-dimensional metric $g$  invariant. The invariance holds, modulo a scale factor, \cite{DMS}:
\begin{equation}  \label{cond1}
g'_{\mu\nu}(x')=F(x)g_{\mu\nu}(x).
\end{equation}
Conformal mappings can be essentially summarized as     Lorentz 
transformations,  translations,  dilations and  special 
conformal transformations. The generators of these transformations are given by the 
operators $M_{\mu\nu}$, $P_{\rho}$, $D$, $K_{\sigma}$.
\\\\
The conformal algebra is defined by the commutation relations of the 
generators and is given as follows:
\begin{equation}
[M_{\mu\nu},M_{\rho\sigma}]
=i\left(\eta_{\mu\sigma}M_{\nu\rho}+\eta_{\nu\rho}M_{\mu\sigma}-\eta_{\mu\rho}M_{\nu\sigma}-
\eta_{\nu\sigma}M_{\mu\rho},
\right)
\end{equation}
\begin{equation}
[P_{\rho},M_{\mu\nu}]=i\left(\eta_{\rho\mu}P_{\nu}-\eta_{\rho\nu}P_{\mu}\right),
\qquad 
[K_{\rho},M_{\mu\nu}]=i\left(\eta_{\rho\mu}K_{\nu}-\eta_{\rho\nu}K_{\mu}\right),
\end{equation}
\begin{equation}\label{f11}
[P_{\rho},D]=iP_{\rho}, \qquad [K_{\rho},D]=-iK_{\rho},
\end{equation}
\begin{equation}\label{f12}
[P_{\rho},K_{\mu}]= 2i\left(\eta_{\rho\mu}D-M_{\rho\mu}\right),
\end{equation}
where all other commutators are equal to zero.
\\\\ 
To transform  the momentum and Lorentz operators in to the NWP-space we can take the massive expressions (see \cite{Muc7})    and perform a massless limit, i.e. $m\rightarrow0$. However, in order to make this work self-contained we  calculate most expressions  explicitly.  
\begin{lemma}
	The zero component in the  massless case is given by,
	\begin{align*}  P_{0}
	=  
	\int
	d^n \mathbf{x}\,
	\tilde{a}^*(\mathbf{x})  \left(\tilde{\omega} \ast
	\tilde{a}  \right)(\mathbf{x} ),
	\end{align*} 
	with   function $\tilde{\omega}(x):=-\pi^{-\frac{n+1}{2}}\Gamma(\frac{n+1}{2})  |\mathbf{x} |^{- ( n+1)}$ and $\ast$ denoting the convolution.\\\\
	The massless spatial momentum operator   takes  the same form in   coordinate space as in the massive case, 
	\begin{align*} 
	P_j&=
	- i
	\int
	d^n \mathbf{x}\,\tilde{a}^*(\mathbf{x})  
	\frac{\partial}{\partial x^{j}}\tilde{a} (\mathbf{x} ).
	\end{align*}
\end{lemma}
\begin{proof}
	For the proof we   take the expression of the massless energy operator (see Equation (\ref{mop})) and transform, by explicit Fourier-transformation, into the coordinate space,
	\begin{align*} 
	P_{0}&=  \int
	d^n \mathbf{p}\,   \omega_{\mathbf{p}} \, {a}^*(\mathbf{p}) {a}(\mathbf{p})\\&=(2\pi)^{-n} 
	\int
	d^n \mathbf{p}\,  |\mathbf{p}| \int
	d^n \mathbf{x}\, e^{ip_kx^k} \tilde{a}^*(\mathbf{x}) \int
	d^n \mathbf{y} \,e^{-ip_ly^l} \tilde{a} (\mathbf{y} ) \\&=(2\pi)^{-n} 
	\iint
	d^n \mathbf{x}\,d^n \mathbf{y}\, \left( \int
	d^n \mathbf{p}\, {|\mathbf{p}| } \, e^{ip_k(x-y)^k} \right)\tilde{a}^*(\mathbf{x})  
	\tilde{a} (\mathbf{y} )	 \\ &= -\pi^{-\frac{n+1}{2}}\Gamma(\frac{n+1}{2})
	\iint
	d^n \mathbf{x}\,d^n \mathbf{y}\, 
	\frac{1}{|\mathbf{x}-\mathbf{y}|^{ {n+1}} } 
	\tilde{a}^*(\mathbf{x})  
	\tilde{a} (\mathbf{y} ),
	\end{align*}
	where the Fourier-transformation can be found in \cite[Chapter III, Section 2.6]{GS1} or see Equation (\ref{f1}). The spatial momentum operator does not depend explicitly on the mass term, hence its representation in the massless  is equivalent to the massive case, (see \cite{Muc7}).
\end{proof}
Other important expressions in the NWP-context are the velocity and the particle number operator. One way to calculate the velocity operator is  by using the Heisenberg equation of motion, 
\begin{align}\label{heq}
[P_{0},X_{j}]=-iV_j.
\end{align}
Since, we represented the relativistic energy and the NWP-operator in coordinate space, the calculation of the commutator gives us the velocity directly in coordinate space. In order to represent the particle number operator in the NWP-space we apply a Fourier-transformation on the expression in momentum space (see Equation (\ref{pn})).
\begin{lemma}
	The velocity operator is given, in the massless case, as follows
	\begin{align}
	V_j= -i\int d^n \mathbf{x} \,  
	\tilde{a}^*(\mathbf{x})  \left(\tilde{\omega} _j\ast
	\tilde{a}\right) (\mathbf{x} ),
	\end{align}
	with function $\tilde{\omega} _j(\mathbf{x}):=-  \pi^{-\frac{n+1}{2}}\Gamma(\frac{n+1}{2}) \, |\mathbf{x} |^{- ( n+1)}\,x_j$ and $\ast$ denoting the convolution. \\\\ Moreover, the particle number operator  $N$ is in the coordinate space given by,
	\begin{align}
	N= \int d^n \mathbf{x} \,  
	\tilde{a}^*(\mathbf{x})  
	\tilde{a}  (\mathbf{x} ).
	\end{align}

\end{lemma}
\begin{proof}
We calculate the velocity by taking the commutator of the spatial coordinate space and the zero component of the momentum (see Heisenberg-Equation (\ref{heq})), which is  given in the former theorem, i.e. 
	\begin{align*}
	[P_0,X_j]&=\iint d^n \mathbf{x} \, d^n \mathbf{z} \,  z_j\, [
	\tilde{a}^*(\mathbf{x}) \left(\tilde{\omega}\ast
	\tilde{a}\right) (\mathbf{x} ),\tilde{a}^*(\mathbf{z}) \tilde{a} (\mathbf{z}) ]\\&=
	\iiint d^n \mathbf{x} \, d^n \mathbf{y}\,d^n \mathbf{z}\,z_j  \, \tilde{\omega}(\mathbf{x}-\mathbf{y})\underbrace{[
		\tilde{a}^*(\mathbf{x})  
		\tilde{a} (\mathbf{y} ),\tilde{a}^*(\mathbf{z}) \tilde{a} (\mathbf{z}) ]}_{-\delta(\mathbf{x}-\mathbf{z}) \tilde{a}^*(\mathbf{z})\tilde{a} (\mathbf{y})+\delta(\mathbf{y}-\mathbf{z}) \tilde{a}^*(\mathbf{x})\tilde{a} (\mathbf{z})}\\&=
	-\iint d^n \mathbf{x} \, d^n \mathbf{y}\,  \,(x-y)_j  \, \tilde{\omega}(\mathbf{x}-\mathbf{y})  \tilde{a}^*(\mathbf{x}) \tilde{a}(\mathbf{y}).
	\end{align*}
	For the particle number operator we use the Fourier-transformation as in the proof of the former lemma,
		\begin{align*} 
	N&=  \int
		d^n \mathbf{p}\,     {a}^*(\mathbf{p}) {a}(\mathbf{p})\\&=(2\pi)^{-n} 
		\int
		d^n \mathbf{p}\,    \int
		d^n \mathbf{x}\, e^{ip_kx^k} \tilde{a}^*(\mathbf{x}) \int
		d^n \mathbf{y} \,e^{-ip_ly^l} \tilde{a} (\mathbf{y} ) \\&=(2\pi)^{-n} 
		\iint
		d^n \mathbf{x}\,d^n \mathbf{y}\,\underbrace{ \left( \int
		d^n \mathbf{p}\,  e^{ip_k(x-y)^k} \right)}_{(2\pi)^{n}\delta(\mathbf{x}-\mathbf{y}) }\tilde{a}^*(\mathbf{x})  
		\tilde{a} (\mathbf{y} ) .
	\end{align*}
\end{proof}
Next, we turn our attention to the expressions of  Lorentz generators.  They are given as in the massive case by (see \cite[Equation
3.54]{IZ} and in this context see also \cite[Appendix]{SS}),
\begin{align}\label{lbcaop1}
M_{j0}&= i\int  d^n\mathbf{p}\,{a}^{*}(\textbf{p}) 
\left(  \frac{p_j}{2\omega_{\textbf{p}}}-\omega_{\textbf{p}}\frac{\partial}{\partial p^j } \right) a(\textbf{p}),
\\ \label{lbcaop2}
M_{ik}&=i \int d^n\mathbf{p}\,  {a}^{*}(\textbf{p}) 
\left(p_i \frac{\partial}{\partial p^k }-p_k\frac{\partial}{\partial p^i }\right)
a(\textbf{p}).
\end{align}
	It is interesting and important to note that   operators of boost and rotations can be given by the second quantization (denoted by $d\Gamma(\cdot)$ see \cite[Chapter X.7]{RS2}) of symmetric or skew-symmetric products of the momentum and coordinate operator.
\begin{theorem}For the massless scalar field the generators of the proper orthochronous Lorentz-group $\mathscr{L}^{\uparrow}_{+}$  can be written in terms of the second quantization of products of the NWP-operator with relativistic four-momentum as,
	\begin{align}\label{p1}
	M_{0j}=\frac{1}{2}\big( d\Gamma(X_jP_0)+d\Gamma(P_0X_j)\big)  ,
\qquad 
		M_{ik} =   d\Gamma(X_iP_k)-d\Gamma(X_{k}P_i).
		\end{align}
	\end{theorem}
	\begin{proof}
	For the boosts we have,
	\begin{align*}
	M_{0j}&=\frac{1}{2}\big( d\Gamma(X_jP_0)+d\Gamma(P_0X_j)\big) =\frac{1}{2}  d\Gamma([X_j,P_0])+d\Gamma(P_0X_j) 
	  \\  &=\frac{i}{2}  d\Gamma(V_j)+d\Gamma(P_0X_j) ,
	\end{align*}
	where in the last lines we used the Heisenberg equation of motion. By comparing the last line with Formula (\ref{lbcaop1}) the proof follows. 	 For rotations, the representation follows trivially. 
	\end{proof}
	The interesting fact about the former Theorem is that in principle one can define the generators of the Lorentz-group in a $k$-particle space by solely using the translation group and the NWP-operator. This resembles the proof of the Maxwell-equations by using the canonical commutation relations of Feynman (published by Dyson,  \cite{FD}). In particular, the existence of the NWP-operator and the translation group are sufficient to build the Poincar\'e group, i.e. the group responsible for the implications of special relativity.
	Hence, the commutation relations with the addition of relativistic energy can be used to implement relativistic principles. This fact, is, in our opinion, an additional  argument for the physical sense of the NWP-operator.   \\\\
Next, we turn our attention to the explicit expressions of  Lorentz generators in the coordinate space.
\begin{lemma}
	The boost generators of the Lorentz group expressed in the terms of ladder operators of the free massless scalar field are given in the coordinate space as
	\begin{align*}
	M_{0j}=\frac{1}{2}	\iint d^n \mathbf{x} \, d^n \mathbf{y}\,(x+y)_j\, \tilde{\omega}(\mathbf{x}-\mathbf{y})  \tilde{a}^*(\mathbf{x}) \tilde{a}(\mathbf{y})  .
	\end{align*}
	The operator of rotations takes  the, from a quantum mechanical point of view, familiar  and to the massive case equivalent form,  
	\begin{align*}
	M_{ik}=	 i \int
	d^n \mathbf{x}  \,
	\tilde{a}^*(\mathbf{x})
	\left(x_i \frac{\partial}{\partial x^k }-x_k\frac{\partial}{\partial x^i }\right) \tilde{a} (\mathbf{x}).
	\end{align*}
\end{lemma}

\begin{proof}Since in the proof of \cite[Lemma 4.3]{Muc7}   we did not explicitly encounter the mass when we calculated the  representation of rotations in the coordinate space, we conclude that they have the same form.	Next, we turn to the proof of the representation of    Lorentz boosts in coordinate space. The first term is simply the massless velocity operator times $\frac{i}{2}$ (see Equation (\ref{p1})). Hence, we focus here only on the second part,
	\begin{align*}
	&- i\int  d^n\mathbf{p}\,{a}^{*}(\textbf{p}) \, \omega_{\mathbf{p}}\frac{\partial}{\partial p^j}  a(\textbf{p})\\&=
	- i\int  d^n\mathbf{p}\,\int
	d^n \mathbf{x}\, e^{-ip_rx^r} \,\tilde{a}^*(\mathbf{x}) \, \omega_{\textbf{p}}\frac{\partial}{\partial p^j}  \int
	d^n \mathbf{y}\, e^{ip_sy^s} \tilde{a}(\mathbf{y}) 
	\\&= (2\pi)^{-n}\iint	d^n \mathbf{x}\,	d^n \mathbf{y}
	\left( \int  d^n\mathbf{p}\, |\mathbf{p}| \, e^{-ip_k(x-y)^k}    \right)\,y_j\,\tilde{a}^*(\mathbf{x}) \tilde{a}(\mathbf{y}) 
	\\&=-\pi^{-\frac{n+1}{2}}\Gamma(\frac{n+1}{2})
	\iint
	d^n \mathbf{x}\,d^n \mathbf{y}\, y_j\, 
	\frac{1}{|\mathbf{x}-\mathbf{y}|^{ {n+1}} } 
	\tilde{a}^*(\mathbf{x})  
	\tilde{a} (\mathbf{y} )
	\\&= 
	\iint
	d^n \mathbf{x}\,d^n \mathbf{y}\, y_j\,\tilde{\omega}(\mathbf{x}-\mathbf{y})
	\tilde{a}^*(\mathbf{x})  
	\tilde{a} (\mathbf{y} )
	,
	\end{align*}  
\end{proof}
In order to translate the special conformal   and the dilatation operator  into the coordinate space, we  first write the momentum space representation, see \cite{SS}. Moreover, we restrict this part to the physical space-time dimension four. The dilatation operator   is given in momentum space  as 
\begin{align}
D=-i\int  d^3\mathbf{p}\,{a}^{*}(\textbf{p}) \left(\frac{3}{2}+p^{l}\frac{\partial}{\partial p^{l}}\right)
{a}(\textbf{p}),
\end{align}
and for the special conformal operators we have,
\begin{align}
K_{0}&=-\int  d^3\mathbf{p}\,{a}^{*}(\textbf{p}) \left(\frac{3}{4 \,\omega_{\mathbf{p}}}+\frac{p^{l}}{\omega_{\mathbf{p}}}\frac{\partial}{\partial p^{l}}-\omega_{\mathbf{p}}\frac{\partial}{\partial p^{l}}\frac{\partial}{\partial p_{l}}\right)
a(\textbf{p}),\\  
K_{j}&= -\int  d^3\mathbf{p}\,{a}^{*}(\textbf{p}) \left(
\frac{p_j}{4\,\omega_{\mathbf{p}}^2}+3 \frac{\partial}{\partial p^{j}}+2p^l\frac{\partial}{\partial p^{l}}\frac{\partial}{\partial p^{j}}-p_{j}\frac{\partial}{\partial p^{l}}\frac{\partial}{\partial p_{l}}\right)
a(\textbf{p}).
\end{align}
\begin{lemma}
	The four-dimensional dilatation operator is given in the coordinate representation as follows,
	\begin{align}
	D=i\int d^3\mathbf{x}\,\tilde{a}^*(\mathbf{x})\left(\frac{3}{2}+x^{l}\frac{\partial}{\partial x^{l}}
	\right)\tilde{a}(\mathbf{x}).
	\end{align}
\end{lemma}
\begin{proof}
	The Fourier-transformation is straight-forward in this case, e.g.
	\begin{align*}
	D&=-i	\int  d^3\mathbf{p}\,{a}^{*}(\textbf{p}) \left(\frac{3}{2}+p^{l}\frac{\partial}{\partial p^{l}}\right)
	{a}(\textbf{p})
	\\&=-i(2\pi)^{-3} 
	\int
	d^3 \mathbf{p}\,    \int
	d^3 \mathbf{x}\, e^{-ip_rx^r} \tilde{a}^*(\mathbf{x}) \left(\frac{3}{2}+p^{l}\frac{\partial}{\partial p^{l}}\right)\int
	d^3 \mathbf{y} \,e^{ ip_ky^k} \tilde{a} (\mathbf{y} ) 	\\&=-
	\frac{3i}{2}\int d^3\mathbf{x}\,\tilde{a}^*(\mathbf{x})\tilde{a}(\mathbf{x})
	-i(2\pi)^{-3} 
	\iint
	d^3 \mathbf{p}\,   
	d^3 \mathbf{x}\, e^{-ip_rx^r} \tilde{a}^*(\mathbf{x})  \int
	d^3 \mathbf{y} \,\left(y^{l}\frac{\partial}{\partial y^{l}}e^{ ip_ky^k}\right) \tilde{a} (\mathbf{y} ) \\&=
	 	\frac{3i}{2}\int d^3\mathbf{x}\,\tilde{a}^*(\mathbf{x})\tilde{a}(\mathbf{x})
	+i(2\pi)^{-3} 
	\iint
	d^3 \mathbf{p}\,  
	d^3 \mathbf{x}\, e^{-ip_rx^r} \tilde{a}^*(\mathbf{x})  \int
	d^3 \mathbf{y} \,e^{ ip_ky^k} y^{l}\frac{\partial}{\partial y^{l}}\tilde{a} (\mathbf{y} ) \\&=
	i	\int d^3\mathbf{x}\,\tilde{a}^*(\mathbf{x})\left(\frac{3}{2}+ x^{l}\frac{\partial}{\partial x^{l}}\right)\tilde{a}(\mathbf{x})  .
	\end{align*}
\end{proof}
The outcome of the dilatation operator is interesting. In essence, it is equivalent to its momentum representation modulo a sign. This is what one might expect, since this operator performs dilatations in a complimentary manner, with regard to the momentum and coordinate space. \\\\
Next, we transform the special conformal operators into the coordinate space. From a calculational point of view these operators are the most difficult ones. 
\begin{lemma}
	The zero component of the  special conformal operators has the following form in coordinate space 
	\begin{align*}
	K_{0}&=  \pi^{-2}	\iint d^3 \mathbf{x}\,d^3 \mathbf{y} \,\tilde{a}^*(\mathbf{x}) \,\vert \mathbf{x}-\mathbf{y}\vert^{-2} \left( \frac{9}{8}-
	\vert \mathbf{x}-\mathbf{y}\vert^{-2}\vert\mathbf{y}\vert^2
	+ \frac{y^{l}}{2} \frac{\partial}{\partial y^{l}}
	\right) \tilde{a} (\mathbf{y} ).
	\end{align*}

\end{lemma}

\begin{proof}
	In order to make this calculation more readable we separate the zero component into three parts as follows,
	\begin{align*}
	K_{0}&=-\int  d^3\mathbf{p}\,{a}^{*}(\textbf{p}) \left(\frac{3}{4 \,\omega_{\mathbf{p}}}+\frac{p^{l}}{\omega_{\mathbf{p}}}\frac{\partial}{\partial p^{l}}-\omega_{\mathbf{p}}\frac{\partial}{\partial p^{l}}\frac{\partial}{\partial p_{l}}\right)
	a(\textbf{p})\\&=:
	K_{0}^1+K_{0}^2+K_{0}^3.
	\end{align*}
	Let us take a look at the first part,
	\begin{align*}
	K^1_{0}&=-\frac{3}{4}\int  d^3\mathbf{p}\,{a}^{*}(\textbf{p})   \,\omega_{\mathbf{p}}^{-1}
	a(\textbf{p})
	\\&=-\frac{3}{4}(2\pi)^{-3} 
	\iint d^3 \mathbf{x}\,d^3 \mathbf{y} \,\tilde{a}^*(\mathbf{x})
	\left( \int d^3 \mathbf{p}   
	\, e^{-ip_k(x-y)^k }  \,\omega_{\mathbf{p}}^{-1}\right)
	\tilde{a} (\mathbf{y} ) 
	\\&= -\frac{3}{8} \pi^{-2}	\iint d^3 \mathbf{x}\,d^3 \mathbf{y} \,\tilde{a}^*(\mathbf{x}) \,\vert \mathbf{x}-\mathbf{y}\vert^{-2}  \tilde{a} (\mathbf{y} ) .
	\end{align*}
	Next we turn our attention to the second term of the zero component of the special conformal operator, 
	\begin{align*}
	K^2_{0}&=-\int  d^3\mathbf{p}\,{a}^{*}(\textbf{p})   \,\frac{p^{l}}{\omega_{\mathbf{p}}}\frac{\partial}{\partial p^{l}}
	a(\textbf{p})
	\\&=- (2\pi)^{-3}\int  d^3\mathbf{p} \int
	d^3 \mathbf{x}\, e^{-ip_rx^r} \tilde{a}^*(\mathbf{x})
	\int
	d^3 \mathbf{y} \,\left(\frac{p^{l}}{\omega_{\mathbf{p}}}\frac{\partial}{\partial p^{l}}e^{ip_ky^k} \right)\tilde{a} (\mathbf{y} ) 
	\\&=  - (2\pi)^{-3} \int
	d^3 \mathbf{x}\,\tilde{a}^*(\mathbf{x})
	\int  d^3\mathbf{p}\,\frac{1}{\omega_{\mathbf{p}}} e^{-ip_rx^r}
	\int d^3 \mathbf{y} \,\left( y^{l}\frac{\partial}{\partial y^{l}}e^{ip_ky^k} \right)\tilde{a} (\mathbf{y} ) 
	\\&=     (2\pi)^{-3} \iint
	d^3 \mathbf{x}\,d^3 \mathbf{y}\,\tilde{a}^*(\mathbf{x})
	\left(\int  d^3\mathbf{p}\,\frac{1}{\omega_{\mathbf{p}}} e^{-ip_k(x-y)^k}\right)
	\, \left(3+y^{l} \frac{\partial}{\partial y^{l}}\right)\tilde{a} (\mathbf{y} )
	\\&=     (2)^{-1}  \pi^{-2} \iint
	d^3 \mathbf{x}\,d^3 \mathbf{y}\,\tilde{a}^*(\mathbf{x})
	\left( \vert \mathbf{x}-\mathbf{y}\vert^{-2}  \right)
	\, \left(3+y^{l} \frac{\partial}{\partial y^{l}}\right)\tilde{a} (\mathbf{y} ),
	\end{align*}
	where in the last lines we performed a partial integration and used Formula (\ref{f1}). The last term of the zero component is given by
	\begin{align*}
	K^3_{0}&=\int  d^3\mathbf{p}\,{a}^{*}(\textbf{p})   \, \omega_{\mathbf{p}}\frac{\partial}{\partial p^{l}}\frac{\partial}{\partial p_{l}}
	a(\textbf{p})\\&
	=(2\pi)^{-3}\int  d^3\mathbf{p} \int
	d^3 \mathbf{x}\, e^{-ip_rx^r} \tilde{a}^*(\mathbf{x})
	\int
	d^3 \mathbf{y} \,\left( 
	\omega_{\mathbf{p}}\,\frac{\partial}{\partial p^{l}}\frac{\partial}{\partial p_{l}}
	e^{ ip_ky^k} \right)\tilde{a} (\mathbf{y} ) 
	\\&=  (2\pi)^{-3}\iint
	d^3 \mathbf{x}\,	   d^3 \mathbf{y}\,\tilde{a}^*(\mathbf{x})\left(
	\int  d^3\mathbf{p}\, {\omega_{\mathbf{p}}}\, e^{-ip_k(x-y)^k}\right) \,\vert\mathbf{y}\vert^2  \tilde{a} (\mathbf{y} ) 
	\\&=  -( \pi)^{-2}\iint
	d^3 \mathbf{x}\,  d^3 \mathbf{y} \,\tilde{a}^*(\mathbf{x})\left(	\vert \mathbf{x}-\mathbf{y}\vert^{-4} 
	\right)	\vert\mathbf{y}\vert^2  \tilde{a} (\mathbf{y} ) , 
	\end{align*}
	where in the last lines we used the action of the differential operators w.r.t. $\mathbf{p}$ on the exponential function and used as before Equation (\ref{f1}).
	
\end{proof}
Next, we give the base change of the spatial part of the special conformal operator. 
\begin{lemma}
	The spatial components of the  special conformal operators have the following form in coordinate space,
	\begin{align}
	K_{j}  =  -\frac{i}{16}\pi^{-1}\iint&
	d^3 \mathbf{x}\,d^3 \mathbf{y}\,\tilde{a}^*(\mathbf{x}) 
	\vert \mathbf{x}-\mathbf{y}\vert^{-1} 
	\, \frac{\partial}{\partial y^{j}}\tilde{a} (\mathbf{y} )\\&\nonumber+
	i \int
	d^3 \mathbf{x} \,\tilde{a}^*(\mathbf{x})
	\left(3x_j+ 2x_{j}x^{l} \frac{\partial}{\partial x^{l}}- x_{l}x^{l} \frac{\partial}{\partial x^{j}} \right) \tilde{a} (\mathbf{x} )
	\end{align}
	
\end{lemma}

\begin{proof}
	Let us first recall and define the expression of the special conformal operators in the momentum space, i.e. 
	\begin{align*}
	K_{j}&= -\int  d^3\mathbf{p}\,{a}^{*}(\textbf{p}) \left(
	\frac{p_j}{4\,\omega_{\mathbf{p}}^2}+3 \frac{\partial}{\partial p^{j}}+2p_l\frac{\partial}{\partial p_{l}}\frac{\partial}{\partial p^{j}}-p_{j}\frac{\partial}{\partial p^{l}}\frac{\partial}{\partial p_{l}}\right)
	a(\textbf{p})\\&=
	K^{1}_{j}-3i	X_{j}+	K^{2}_{j}+	K^{3}_{j}.
	\end{align*}
	The second term in the spatial conformal operator is simply the coordinate operator times a constant. This fact was used in \cite{SE}. Hence, the remaining terms of interest are $K^{1}_{j},\,K^{2}_{j},\,K^{3}_{j}$. Let us start with first object, 
	\begin{align*}
	K^{1}_{j}&= -
	\frac{1}{4}
	\int  d^3\mathbf{p}\,{a}^{*}(\textbf{p}) \left(
	\frac{p_j}{\,\omega_{\mathbf{p}}^2} \right)
	a(\textbf{p})\\&= -
	\frac{1}{4}
	(2\pi)^{-3}\int  d^3\mathbf{p} \int
	d^3 \mathbf{x}\, e^{-ip_rx^r} \tilde{a}^*(\mathbf{x})
	\int
	d^3 \mathbf{y} \,\left(	\frac{p_j}{\,\omega_{\mathbf{p}}^2} e^{ ip_ky^k} \right)\tilde{a} (\mathbf{y} ) \\&=   -
	\frac{i}{4}
	(2\pi)^{-3} \iint
	d^3 \mathbf{x}\,d^3 \mathbf{y}\,\tilde{a}^*(\mathbf{x})
	\left(\int  d^3\mathbf{p}\,\frac{1}{\omega_{\mathbf{p}}^2} e^{-ip_k(x-y)^k}\right)
	\, \frac{\partial}{\partial y^{j}}\tilde{a} (\mathbf{y} )
	\\&=  - \frac{i}{16}\pi^{-1}\iint
	d^3 \mathbf{x}\,d^3 \mathbf{y}\,\tilde{a}^*(\mathbf{x}) 
	\vert \mathbf{x}-\mathbf{y}\vert^{-1} 
	\, \frac{\partial}{\partial y^{j}}\tilde{a} (\mathbf{y} ),
	\end{align*}
	where in the last lines we expressed the momentum as a derivative, performed a partial integration and used the Fourier-transformation given in Formula (\ref{f1}) for $\lambda=-1$. A  general term which includes, after taking the corresponding indices, the second and third operator of the spatial special conformal object is given by,
	\begin{align*}
	O_{j}&= 
	\int  d^3\mathbf{p}\,{a}^{*}(\textbf{p}) \left(p_{r}\frac{\partial}{\partial p_{l}}\frac{\partial}{\partial p^{j}}\right)
	a(\textbf{p})\\&= 
	(2\pi)^{-3}\int  d^3\mathbf{p} \int
	d^3 \mathbf{x}\, e^{-ip_sx^s} \tilde{a}^*(\mathbf{x})
	\int
	d^3 \mathbf{y} \,\left(	p_r\frac{\partial}{\partial p_{l}}\frac{\partial}{\partial p^{j}}e^{ ip_ky^k} \right)\tilde{a} (\mathbf{y} )
	\\&=  i
	(2\pi)^{-3}\int  d^3\mathbf{p} \int
	d^3 \mathbf{x}\, e^{-ip_sx^s } \tilde{a}^*(\mathbf{x})
	\int
	d^3 \mathbf{y} \,\left(y_{j}y^{l} \frac{\partial}{\partial y^{r}} e^{ ip_ky^k} \right)\tilde{a} (\mathbf{y} )\\&= - i \int
	d^3 \mathbf{x} \,\tilde{a}^*(\mathbf{x})
	\left(\eta_{rj}x^l+\eta_{r}^{\,\,l}x_{j}+ x_{j}x^{l} \frac{\partial}{\partial x^ {r}} \right) \tilde{a} (\mathbf{x} ),
	\end{align*}
	where in the last lines we applied the derivatives to the exponential, expressed the coordinate $y$	as a derivative and performed a partial integration.   
\end{proof}
  \begin{remark}
  	As for  Lorentz generators,  dilatation and special conformal operators can be written in terms of symmetric second-quantized products of the momentum and NWP-operator. For the dilatation operator it is straight forward,
  	\begin{align*}
  	D=\frac{1}{2}\big(d\Gamma(P_jX^j) +d\Gamma(X^jP_j)\big).
  	\end{align*}
  	while the special-conformal operators are a bit more involved. For the zero component we have, 
  		\begin{align*}
  		K_0&=-\frac{3}{4}d\Gamma(P_0^{-1})-id\Gamma(V^{l}X_l)+d\Gamma(P_0X_lX^l) 	\end{align*}	
while the spatial components in terms of the NWP and the momentum operator read, 
\begin{align*}
  		K_j&=-\frac{1}{4}d\Gamma(P_0^{-1}V_j)-3id\Gamma(X_j)+2d\Gamma(P^lX_lX_j)-d\Gamma(P_jX_lX^l) .
  		\end{align*}
  		 
  \end{remark}
   
  In this section we expressed all generators of the conformal group in terms of the ladder operators of the NWP space for $x_0=0$. However, the operator $\phi_1$ (see Equation (\ref{copes}))  responsible for calculating the probability of finding $k$-particle at a certain position has an explicit time dependence. Hence, in order  to connect our results with this time-dependency we give the following proposition.
  \begin{proposition}\label{sf}
The operator $\phi_1$ that has the following form,
  	\begin{align*} 
  	\phi_1( {x}) =(2\pi)^{-n/2}\,
  	\int d^n \mathbf{p}\, e^{-i {p} \cdot {x}}\,a(\mathbf{p}),
  	\end{align*}
 can be obtained by performing a time-translation on the   annihilation operator $\tilde{a}$ of the NWP-space
  	as follows
  	\begin{align*} 
  	\phi_1( x^{0}, \mathbf{x})  = e^{ix^{0}P_{0}} \tilde{a}(	 \mathbf{x})e^{-ix^{0}P_{0}}, \qquad x^0\in\mathbb{R}.
  	\end{align*}
  	
  \end{proposition}	\begin{proof}
  By using the inverse Fourier-transformation  of $\tilde{a}$ and the transformation of the annihilation operator in momentum space under time-translation, the proof is completed, i.e. 
  \begin{align*} 
  e^{ix^{0}P_{0}} \tilde{a}(	 \mathbf{x})e^{-ix^{0}P_{0}}&=(2\pi)^{-n/2} \int
  d^n \mathbf{p}\, e^{-ip_kx^k}e^{ix^{0}P_{0}}{a}(\mathbf{p})e^{-iy^{0}P_{0}}\\&=(2\pi)^{-n/2} \int
  d^n \mathbf{p}\,  e^{-ip_kx^k}\,e^{-ip_{0}x^{0} }{a}(\mathbf{p}) .
  \end{align*}
\end{proof}Hence, we can simply time-translate all the operators of the conformal group in the NWP-space, by using the same time $x^0$, in order to obtain the expressions for a general time and in terms of the operator $\phi_1$. Since the zero component of the momentum commutes with the spatial momentum and the rotations, we only have to calculate the time-translation for the remaining generators by using the algebra.

\begin{lemma}
	The dilatation transforms under time-translations as follows,
	\begin{align*}
	e^{ix^{0}P_{0}}D e^{-ix^{0}P_{0}} 	= D-x^{0}P_{0}, 
	\end{align*}
	where $x^0\in\mathbb{R}$. Under time-translations the special conformal operator transforms as,
	\begin{align*}
	e^{ix^{0}P_{0}}K_{\mu} e^{-ix^{0}P_{0}}&=
	K_{\mu}-2x^{0}\left(
	\eta_{0\mu}D-M_{0\mu}
	\right)-(x^{0})^2
	\left( P_{\mu}\right).
	\end{align*}
\end{lemma}
\begin{proof}
	The proof solely uses the conformal algebra (see Formulas (\ref{f11}) and (\ref{f12})).  For the dilatations   we have,
	\begin{align*}
	e^{ix^{0}P_{0}}D e^{-ix^{0}P_{0}} =
	D-x^{0}P_{0} +\frac{i^2}{2!}(x^{0})^2
	\underbrace{[P_{0},[P_{0},D]]+ \cdots}_{=0}= D-x^{0}P_{0},
	\end{align*}
	where in the last lines we used the Backer-Hausdorff formula.	A more complex expression is the transformation of the special-conformal operator, i.e. 
	\begin{align*}
	e^{ix^{0}P_{0}}K_{\mu} e^{-ix^{0}P_{0}}&=K_{\mu}+ix^{0}[P_{0},K_{\mu}]+\frac{i^2}{2!}(x^{0})^2
	[P_{0},[P_{0},K_{\mu}]]+\underbrace{\cdots}_{=0}
	\\&= K_{\mu}-2x^{0}\left(
	\eta_{0\mu}D-M_{0\mu}
	\right)-(x^{0})^2
	\left( \eta_{00}P_{\mu}\right) ,
	\end{align*}
	where we used the Baker-Campbell-Hausdorff formula and the specific commutator relations of the   conformal algebra.
	
\end{proof}
Since, we have all expressions in terms of generators of the NWP-space we can equivalently give the conformal algebra expressed by  operators $\phi^*_1(x_0,\mathbf{x})$ and $\phi_1(x_0,\mathbf{y})$. Note that neither the coordinate operator nor its respective  eigenstates are covariant w.r.t. general Lorentz-transformations. Nevertheless, the generators of the conformal group obey the covariant transformation property independent of their representation. 	Hence, even though we use non-covariant eigenstates, operators of the conformal group represented in the NWP-space respect the Poincar\'e symmetry. 
\section{Conformal-Transformations of NWP-States} 
In this section we  calculate  the explicit adjoint  action of the  conformal group on the Fourier-transformed ladder operators. In some cases this is done by taking the massless limit of the massive theory. \\\\
	 The main physical motivation to calculate the adjoint action w.r.t. conformal transformations, in the coordinate space, is given by the fact that we   use  operator   $\phi_1(x)$  to calculate the probability   of localizing  $k$-particles at $k$-spatial positions. Therefore,  we intend to investigate under which subgroup of the conformal group the probability amplitudes of localization are covariant.  
	\\\\ In order find the covariance group, we first define the unitary operator generating  transformations  of the orthochronous proper Poincar\'e group  $\mathscr{P}^{\uparrow}_{+}=\mathscr{L}^{\uparrow}_{+}\ltimes\mathbb{R}^4$ denoted by $U(y,\Lambda)$. It transforms  creation and annihilation operators in the following fashion, \cite[Chapter 7]{Sch},
	\begin{align}\label{traf}
	U(y,\Lambda) a(\mathbf{p})U(y,\Lambda)^{-1}&= \sqrt{\frac{	\omega_{\Lambda \mathbf{p}}}{	\omega_{ \mathbf{p}}}}e^{-i(\Lambda p)_{\mu}y^{\mu}}{a}(\Lambda\mathbf{p}) ,\qquad (y,\Lambda) \in \mathscr{P}^{\uparrow}_{+},\\\label{traf1}
	U(y,\Lambda) a^{*}(\mathbf{p})U(y,\Lambda)^{-1}&= \sqrt{\frac{	\omega_{\Lambda \mathbf{p}}}{	\omega_{ \mathbf{p}}}}e^{ i(\Lambda p)_{\mu}y^{\mu}}{a}^{*}(\Lambda\mathbf{p}) ,\qquad (y,\Lambda) \in \mathscr{P}^{\uparrow}_{+}
	.
	\end{align} 
	By using the former equations we first calculate the adjoint action of translations on the NWP-ladder operators.
	\begin{lemma}\label{l63}
		The Fourier-transformed ladder  operator  $\tilde{a}$   of the massless field  transforms under  translations as follows,  
		\begin{align}\label{eq51}
		U({y},\mathbb{I})  \tilde{a} (	 \mathbf{x}) U({y},\mathbb{I}) ^{-1} = c_{y } \left(
		\frac{	 1}{\vert \mathbf{x+y} \vert^{2}-(y^{0})^2 } \right)^2\ast\tilde{a} (\mathbf{x+y}),
		\end{align} 
		where $c_{y }= \frac{i y^{0}}{ \pi^2}$, $\ast$ denotes the convolution and $y\in\mathbb{R}^d$.
	\end{lemma}
	\begin{proof}
		The proof for spatial translations can be done as in the massive case (see \cite[Lemma 5.1]{Muc7})  since for this proof there is no explicit mass dependency. In particular they act accordingly as spatial translations in coordinate space, i.e.
			\begin{align*}		U( \mathbf{y},\mathbb{I})  \tilde{a} (	 \mathbf{x}) U( \mathbf{y},\mathbb{I}) ^{-1} 
			 &=
			(2\pi)^{3/2} \int
			d^{3} \mathbf{p}\, e^{-ip_kx^k}    U(\mathbf{y},\mathbb{I}) {a}(\mathbf{p})  U(\mathbf{y},\mathbb{I}) ^{-1}\\
			&=(2\pi)^{-3/2} \int
			d^{3} \mathbf{p}\, e^{-ip_k(x+y)^k}     {a}(\mathbf{p})= \tilde{a}(	 \mathbf{x}+\mathbf{y}).
				\end{align*} 
		 For  time translations, we  take the massless limit of the massive case (see \cite[Lemma 5]{Muc7}) where  we had the integral, 
			\begin{align*}(2\pi)^{-3}\int
			d^{3} \mathbf{p}\, e^{-ip_kx^k}   e^{-iy^{0}  \omega_{m,\mathbf{p}}} =
		\frac{i	y^{0}}{ 2\pi^2} \frac{ m^2K_{2}(m\sqrt{\vert \mathbf{x}-\mathbf{z}\vert^{2}-(y^{0})^2})}{\vert \mathbf{x}-\mathbf{z}\vert^{2}-(y^{0})^2},\end{align*} 
		 with  $\omega_{m,\mathbf{p}}$ being  the energy for a   particle with mass $m$. By using the following limit, 
		\begin{align}\label{ml0}
		\lim\limits_{m\rightarrow 0} m^2\, K_{2}(m f)=\frac{2}{f^2},
		\end{align} 
	 the explicit form of the time translation can be given in the massless case as follows, 
			 \begin{align*} U(y^{0},\mathbb{I}) \tilde{a}(	 \mathbf{x})  U(y^{0},\mathbb{I})^{-1} &=e^{iy^{0}P_{0}} \tilde{a}(	 \mathbf{x})e^{-iy^{0}P_{0}}\\
			 &=(2\pi)^{-3/2} \int
			 d^{3} \mathbf{p}\, e^{-ip_kx^k}   e^{-iy^{0}  \omega_{\mathbf{p}}}   {a}(\mathbf{p})  \\
			 &=(2\pi)^{-3 }\int
			 d^{3} \mathbf{z}    {\left(\int
			 	d^{3} \mathbf{p}\, e^{-ip_k(x-z)^k} e^{-i  \omega_{\mathbf{p}}y^{0}}\right)}_{
			 }  \tilde{a}(\mathbf{z})\\&=\frac{i	y^{0}}{ \pi^2}
			 \int
			 d^{3} \mathbf{z}   \left(
			 \frac{	 1}{\vert \mathbf{x-z} \vert^{2}-(y^{0})^2 } \right)^2\tilde{a}(\mathbf{z}).
			 \end{align*}
			 By taking the commutativity of the components of the momentum operator into account, we obtain the general transformation behavior under translations by using the former two results,  i.e.
			 	\begin{align} 	U(  {y},\mathbb{I})  \tilde{a} (	 \mathbf{x}) U(  {y},\mathbb{I}) ^{-1} =
		  U(y^{0},\mathbb{I}) 	U( \mathbf{y},\mathbb{I})  \tilde{a} (	 \mathbf{x}) U( \mathbf{y},\mathbb{I}) ^{-1}   U(y^{0},\mathbb{I})^{-1}.
			 	\end{align} 
			 
	\end{proof}
Before we turn to the  Lorentz transformations in the massless case, we define the rotation group. 	In particular we denote   matrices of the Lorentz group representing  pure rotations as  $\Lambda_R$. They are given by matrices of the following form,
\[\Lambda_R=\left(
\begin{matrix} 
1 & 0\\
0 & R	 
\end{matrix}\right), \qquad R\in SO(3).
\]
	\begin{lemma}\label{l622}Under pure rotations the   Fourier-transformed ladder  operators $\tilde{a}$   transforms in a covariant manner  as follows, 
		\begin{align*}
		U(0,\Lambda_R)  \tilde{a} (	 \mathbf{x}) U(0,\Lambda_R) ^{-1} =
		\tilde{a} (	 \mathbf{R x }) .
		\end{align*}
	 
	\end{lemma}
	\begin{proof}
		The proof is done analog to the proof of the former Lemma. However, we could take the results of \cite[Lemma 5.2]{Muc7}  and perform the massless limit by using Equation (\ref{ml0}).  First, we calculate the action of pure spatial rotations  on the   Fourier-transformed annihilation operator, 
			\begin{align*}
			U(0,\Lambda_R)  \tilde{a}^{\,}(	 \mathbf{x})  	U(0,\Lambda_R)  ^{-1} &=
			(2\pi)^{-3/2} \int
			d^{3} \mathbf{p}\, e^{-ip_kx^k}   U(0,\Lambda_R){a}(\mathbf{p}) U(0,\Lambda_R)^{-1}\\&=
			(2\pi)^{-3/2} \int
			d^{3} \mathbf{p}\, e^{-ip_kx^k}   \sqrt{\frac{	\omega_{R\mathbf{p}}}{	\omega_{ \mathbf{p}}}}{a}(R\mathbf{p})  
			\\&=
			(2\pi)^{-3/2} \int
			d^{3} \mathbf{p}\, e^{-i p_k(R x)^k}    {a}( \mathbf{p}),
			\end{align*}
			in the last lines we used the transformation behavior of the non-covariant momentum ladder operators  (see \cite[Chapter 7, Equation 57]{Sch}) and the orthogonality of $\Lambda_{R}$ for pure spatial rotations. 
	\end{proof}	
The NWP ladder operators transform non-covariantly w.r.t. the Lorentz-boosts. This is to be expected, since the NWP-operator is non-covariant w.r.t. the boosts. However, we see in the next theorem that the operator $\phi_1(x)$ is at least covariant w.r.t. space-time translations and pure rotations.
\begin{theorem}\label{t2}
	The NWP-ladder operators transform    under space-time translations and pure rotations, i.e. 	with $(y,\Lambda_{R}) \in \mathscr{P}^{\uparrow}_{+}$, as follows,
\begin{align}
	U(y,\Lambda_R)  \tilde{a}(\mathbf{x}) U(y,\Lambda_R)^{-1}&=
	 \frac{i y^{0}}{ \pi^2} \left(
	 \frac{	 1}{\vert \mathbf{Rx+y} \vert^{2}-(y^{0})^2 } \right)^2\ast\tilde{a} (\mathbf{Rx+y}) \\
	 	U(y,\Lambda_R)  \tilde{a}^{*}(\mathbf{x}) U(y,\Lambda_R)^{-1}&=-
	 	\frac{i y^{0}}{ \pi^2} \left(
	 	\frac{	 1}{\vert \mathbf{Rx+y} \vert^{2}-(y^{0})^2 } \right)^2\ast\tilde{a}^{*} (\mathbf{Rx+y}).
\end{align}
Moreover, the operator used to calculate the probabilities, i.e. $\phi_1(x)$ transforms covariantly,
\begin{align}\label{t2e2}
U(y,\Lambda_R) \phi_1(x)   U(y,\Lambda_R)^{-1}= \phi_1(\Lambda_Rx+y)\qquad (y,\Lambda_{R}) \in \mathscr{P}^{\uparrow}_{+}.
\end{align}

\end{theorem}

\begin{proof}
	By using Lemmas \ref{l63} and \ref{l622} we deduce,
		\begin{align*}   	U(y,\Lambda_R)  \tilde{a}(\mathbf{x}) U(y,\Lambda_R)^{-1}  &=  	U(y,\mathbb{I})	U(0,\Lambda_R)\tilde{a}(\mathbf{x})U(0,\Lambda_R)^{-1}	U(y,\mathbb{I})^{-1}		
	\\&=		 	  	U(y,\mathbb{I}) \tilde{a}(R\mathbf{x})
		U(y,\mathbb{I})^{-1}	 
		\\&=   \frac{i y^{0}}{ \pi^2} \left(
		\frac{	 1}{\vert \mathbf{Rx+y} \vert^{2}-(y^{0})^2 } \right)^2\ast\tilde{a} (\mathbf{Rx+y}) ,	
		\end{align*} 
		where in the last lines we used the fact that a general Poincar\'e transformation can be split in first Lorentz-transformation and then space-time translation. An analog proof can be done for the creation operator $\tilde{a}^{*}$ or we can use the fact that we have unitary operators generating the transformations and take the adjoint of the former equality. \\\\ In order to calculate the adjoint action of the operator $\phi_1(x)$, we use the fact that time translations commute with the unitary operator of the Poincar\'e group for pure rotations, i.e.
		\begin{align*}
			U(y,\Lambda_R) \phi_1(x)   U(y,\Lambda_R)^{-1}& =
				U(y,\Lambda_R)U(x_0,\mathbb{I}) \tilde{a}(\mathbf{x})   U(x_0,\mathbb{I})^{-1}U(y,\Lambda_R)^{-1}\\&=U(x_0,\mathbb{I})
				U(y,\Lambda_R) \tilde{a}(\mathbf{x})   U(y,\Lambda_R)^{-1}U(x_0,\mathbb{I})^{-1}\\&=
				U(x_0+y_0,\mathbb{I})
				U(\mathbf{y},\mathbb{I}) \tilde{a}(R\mathbf{x})   U(\mathbf{y},\mathbb{I})^{-1}U(x_0+y_0,\mathbb{I})^{-1}
				\\&=
				U(x_0+y_0,\mathbb{I})  \tilde{a}(R\mathbf{x+y})    U(x_0+y_0,\mathbb{I})^{-1}= \phi_1(\Lambda_Rx+y).
		\end{align*}

\end{proof}
Next, we investigate the  adjoint action w.r.t. the dilatations. 
	In order to  calculate the adjoint action of   dilatations on the NWP-space we first calculate the commutator of the dilatation operator with the coordinate operator. From that we can deduce the adjoint action of the dilatations acting on the coordinate ladder operators. This construction is analog to the one  in  \cite[Chapter 9.4]{SU}.
	\begin{lemma}
The adjoint action w.r.t. the dilatations acting  on  the NWP-operator is given as follows,
	\begin{align*}
	e^{i\alpha D} X_j e^{-i\alpha D}=e^{-\alpha}X_j,
	\end{align*}
where in order to perform the calculations of the adjoint action the following commutator relation was used,
	\begin{align*}
	[D,X_{j}]=iX_{j}.
	\end{align*} 
From the former equation it follows that $X_{j}e^{-i\alpha D}|\mathbf{x}\rangle$ is an eigenvector of 
$X_{j}$ with eigenvalue $e^{-\alpha}x_{j}$. Therefore, one concludes that the transformation behavior of the operator $\tilde{a}(\mathbf{x})$  under dilatations is
	\begin{align*}
	e^{i\alpha D}\tilde{a}(\mathbf{x})	e^{-i\alpha D}=
e^{ \frac{3}{2}\alpha}\,\tilde{a}(e^{ \alpha}\mathbf{x}).
	\end{align*}
 
	\end{lemma}

	\begin{proof}
		In order to proof the transformational behavior we first calculate the commutator of the dilatation operator with the NWP-operator, 
		\begin{align*}
		[D,X_{j}]&=
		i	\iint d^3\mathbf{x}\,d^3\mathbf{y}\,y_j\,[\tilde{a}^*(\mathbf{x})\left(\frac{3}{2}+ x^{l}\frac{\partial}{\partial x^{l}}\right)\tilde{a}(\mathbf{x}) ,  \tilde{a}^*(\mathbf{y}) \tilde{a}(\mathbf{y})]\\&=
		i	\iint d^3\mathbf{x}\,d^3\mathbf{y}\,y_j\,x^{l}\,\underbrace{[\tilde{a}^*(\mathbf{x})  \frac{\partial}{\partial x^{l}} \tilde{a}(\mathbf{x}) ,  \tilde{a}^*(\mathbf{y}) \tilde{a}(\mathbf{y})]}_{-\delta(\mathbf{x}-\mathbf{y})\tilde{a}^*(\mathbf{y})\frac{\partial}{\partial x^{l}} \tilde{a}(\mathbf{x})+\tilde{a}^*(\mathbf{x}) \tilde{a}(\mathbf{y})\frac{\partial}{\partial x^{l}}\delta(\mathbf{x}-\mathbf{y})} 	 
			\\&	=-i\int d^3x \left(x_{j}\,x^l \tilde{a}^*(\mathbf{x})\frac{\partial}{\partial x^{l}}\tilde{a}(\mathbf{x})+
			  \,x_j\,\left(\frac{\partial}{\partial x^{l}}\left(x^l\tilde{a}^*(\mathbf{x})\right) \right) \tilde{a}(\mathbf{x})\right)
				\\&=		-i\int d^3x \left(x_{j}\,x^l \tilde{a}^*(\mathbf{x})\frac{\partial}{\partial x^{l}}\tilde{a}(\mathbf{x})-x^l\tilde{a}^*(\mathbf{x})
			\left(\frac{\partial}{\partial x^{l}}\left(	 x_j\, \tilde{a}(\mathbf{x})\right) \right)\right)
				\\&=	 i\int d^3x \,x_j\, 	\tilde{a}^*(\mathbf{x})\tilde{a}(\mathbf{x}),
		\end{align*}
		where in the last lines we used the vanishing commutator of the particle number operator with the coordinate operator. Moreover, we integrated over delta distributions and performed partial integrations. \\\\		From the former commutator relation we obtain for the adjoint action of the NWP-operator w.r.t.   dilatations the following, 
		\begin{align*}
		 e^{i\alpha D}X_{j}e^{-i\alpha D}&=X_{j}+i\alpha [D,X_{j}]+\frac{i^2}{2!}\alpha^2 [D,[D,X_{j}]]+\cdots\\
		&= X_{j}+i^2\alpha\,X_{j}+\frac{i^4}{2!}\alpha^2X_{j}+\cdots
		\\&= \sum\limits_{n=0}^{\infty}\frac{(i^2\alpha)^n}{n!}X_{j}=e^{-\alpha}X_j,
		\end{align*}
		where in the last lines we used the commutator relation between the dilatation and the NWP-operator and the Baker-Campbell-Hausdorff formula. As in  \cite[Chapter 9.4]{SU} we look at the following expression,
			\begin{align*}
			  X_{j}\,e^{-i\alpha D}|\mathbf{x}\rangle &=e^{-i\alpha D}e^{i\alpha D} X_{j}\,e^{-i\alpha D}|\mathbf{x}\rangle\\&= e^{-i\alpha D}e^{-\alpha}X_j|\mathbf{x}\rangle=e^{-\alpha}x_j\,e^{-i\alpha D}|\mathbf{x}\rangle.
			\end{align*}
		Therefore $X_{j}e^{-i\alpha D}|\vec{x}\rangle$ is an eigenvector of $X_{j}$ for
		the eigenvalue $e^{-\alpha}x_j$, 
		thus contained in the
		eigenspace $\mathcal{H}_{e^{-\alpha}\vec{x} }$ of $X_{j}$. Hence,   the transformation on the creation operator of the NWP-space has to have the following form, 
		\begin{align}\label{daf}
		e^{i\alpha D}\tilde{a}^*(\mathbf{x})	e^{-i\alpha D}=
		f(\alpha)\,\tilde{a}^*(e^{ \alpha}\mathbf{x}),
		\end{align}
		where $f(\alpha)$ is a real-valued function of the transformation parameter $\alpha$. To find the concrete form of the function $f$ we calculate the adjoint action of the coordinate operator
		w.r.t. the dilatations but in the  Fock representation of the NWP-space, i.e. 
	\begin{align*}
	e^{i\alpha D}X_j	e^{-i\alpha D}&= e^{i\alpha D}\,\int d^3 \mathbf{x}\,x_j\,
	\tilde{a}^*( \mathbf{x})	\tilde{a} ( \mathbf{x})	e^{-i\alpha D}\\&=
	\int d^3 \mathbf{x}\,x_j\,
	f(\alpha)^2\tilde{a}^*(e^{ \alpha}\mathbf{x})\tilde{a} (e^{ \alpha}\mathbf{x})\\&=e^{-\alpha}
	\int d^3 \mathbf{x}\,e^{ -3\alpha }\,x_j\,
	f(\alpha)^2\tilde{a}^*( \mathbf{x})\tilde{a} ( \mathbf{x}), 
	\end{align*}
		where in order for the NWP-operator to transform correctly the function $f$ has to be chosen as  $f(\alpha)= e^{ \frac{3}{2}\alpha }$. In order to solidify  our approach   we apply the adjoint action on the  particle number operator   in the coordinate space, 
		\begin{align*}
 e^{i\alpha D}N	e^{-i\alpha D}&= e^{i\alpha D}\,\int d^3 \mathbf{x}\,
	\tilde{a}^*( \mathbf{x})	\tilde{a} ( \mathbf{x})	e^{-i\alpha D}\\&=
	 \int d^3 \mathbf{x}\,
		f(\alpha)^2\tilde{a}^*(e^{ \alpha}\mathbf{x})\tilde{a} (e^{ \alpha}\mathbf{x})\\&=
		\int d^3 \mathbf{x}\,e^{ -3\alpha }
		f(\alpha)^2\tilde{a}^*( \mathbf{x})\tilde{a} ( \mathbf{x})=N, 
		\end{align*}
	where in the last lines a variable substitution was performed and the commutator relation $[D,N]=0$ is used to argue that the adjoint action leaves the particle number operator invariant and therefore the Ansatz (\ref{daf}) with explicit form of $f$  gives us the right transformational behavior. 
	\end{proof}
	 To calculate the adjoint action of dilatations on the ladder operators of the NWP-space we could have, as well, taken the well-known transformation property of the ladder operators in   momentum space. After doing so we could have performed an inverse Fourier-transformation. However, one of the purposes of this work is to demonstrate to the reader that all calculations can as well be done  in  the NWP-space. In particular, the NWP-space is on an equal footing with the momentum space. \\\\
Since we are interested in transformations of the  time translated expression of the ladder operators we need to take the commutation relations between the momentum and dilatation operators into account. 
\begin{theorem}\label{t3}
	The  operator  $\phi_1(x)$ that is responsible for calculating probabilities of finding $k$-particles in $k$-spatial positions   transforms under dilatations  as follows, 
	\begin{align*}
e^{i\alpha D}\phi_1(x)e^{-i\alpha D}=e^{\frac{3}{2} \alpha}\phi_1(e^{\alpha}x).
	\end{align*}
\end{theorem}
\begin{proof}
We can  rewrite the operator  $\phi_1(x)$ as 
	\begin{align*}
\phi_1(x)= 	U(x^{0},\mathbb{I}) \tilde{a} (	 \mathbf{x}) U(-x^{0},\mathbb{I}).
	\end{align*}
Since we know the adjoint action on the ladder operators in the NWP-space we can directly calculate the transformation on all expressions. First we calculate the adjoint action of the dilatations on the zero component of the momentum operator,
	\begin{align*}
e^{i\alpha D}\,P_0\,e^{-i\alpha D}&= -\pi^{-2}\Gamma(2)
\iint
d^3 \mathbf{x}\,d^3 \mathbf{y}\, 
\frac{1}{|\mathbf{x}-\mathbf{y}|^{ 4} } 
e^{i\alpha D}\,\tilde{a}^*(\mathbf{x}) \, e^{i\alpha D}e^{-i\alpha D}\,
\tilde{a} (\mathbf{y} )\,e^{-i\alpha D}\\&
= -\pi^{-2}\Gamma(2)
\iint
d^3 \mathbf{x}\,d^3 \mathbf{y}\, 
\frac{1}{|\mathbf{x}-\mathbf{y}|^{ 4} } \, e^{ 3\alpha}\,
 \tilde{a}^*(e^{ \alpha}\mathbf{x})  \,
\tilde{a} (e^{ \alpha}\mathbf{y} ) \\&
= e^{\alpha} P_{0},
	\end{align*}
where in the last lines we used the explicit transformation given in Equation (\ref{daf}) and performed a variable substitution.  Next,  we apply the adjoint action of   dilatations on this operator and use the fact that we are dealing with unitary transformations, i.e.
\begin{align*}
e^{i\alpha D}U(x^{0},\mathbb{I}) \tilde{a} (	 \mathbf{x}) U(-x^{0},\mathbb{I})e^{-i\alpha D} &=e^{\frac{3}{2} \alpha}
U(e^{\alpha}x^{0},\mathbb{I}) \tilde{a} (e^{ \alpha}  \mathbf{x}) U(-e^{\alpha}x^{0},\mathbb{I})\\&= \frac{i	e^{\frac{5}{2} \alpha}  x^{0}}{2\pi^2}
\int
d^{3} \mathbf{z}   \left(
\frac{	 1}{\vert  e^{ \alpha}\mathbf{x-z} \vert^{2}-(e^{\alpha}x^{0})^2 } \right)^2\tilde{a}(\mathbf{z}).
\end{align*}
\end{proof}
By using the former lemmas we calculate the adjoint action of translations, rotations and dilatations on $\phi_1(x)$. This is in particular important since the operator   $\phi_1(x)$ is used to calculate probability amplitudes of finding $k$-particles in $k$-spatial positions. Hence, by knowing the transformational behavior, two frames that are connected via translations, rotations and/or dilatations can calculate the same probability by taking the proper transformation into account.
	\begin{theorem}
		The  operator	  $\phi_1(x)$ transforms covariantly  under space-\textbf{time} translations,  pure rotations and dilatations, i.e. 
		\begin{align*}&
	e^{i\alpha D}	U(y,\Lambda_R) \phi_1(x)U(y,\Lambda_R)^{-1}	e^{-i\alpha D}
		= e^{\frac{3}{2} \alpha}\phi_1(	e^{\alpha}(\Lambda_R x+ y))
  ,
		\end{align*} 
		with $(y,\Lambda_{R}) \in \mathscr{P}^{\uparrow}_{+}$ and $\alpha\in\mathbb{R}$.
		Moreover, the explicit transformation   of the Fourier-transformed operator $\tilde{a}$ under the action of the subgroup $(y,\Lambda_{R}) \in \mathscr{P}^{\uparrow}_{+}$ is in coordinate space given as, 
		\begin{align*}&
		e^{i\alpha D}		U(y,\Lambda_R)\tilde{a}(\mathbf{x})U(y,\Lambda_R)^{-1}e^{-i\alpha D}
		= c_{y' } \left(
		\frac{	 1}{\vert  \mathbf{Rx}+ \mathbf{y} \vert^{2}-( y^{0})^2 } \right)^2\ast\tilde{a} (e^{\alpha}(\mathbf{Rx}+ \mathbf{y})  ),
		\end{align*} 
		where $c_{y' }:= \frac{i  e^{-\frac{3}{2}\alpha} y^{0}}{ \pi^2}$, $\ast$ denotes the convolution and $y\in\mathbb{R}^4$.
 
	\end{theorem}
\begin{proof}
In order to prove the first part of the theorem we use 	the former  Lemmas and the algebra of the Poincar\'e group. First, we examine how the operator $\phi_1(x)$ transforms under space-time translations and pure rotations,
	\begin{align*}
	&  e^{i\alpha D}		U(y,\Lambda_R) \phi_1(x)U(y,\Lambda_R)^{-1}	e^{-i\alpha D}\\ & =	e^{i\alpha D}	\phi_1(\Lambda_R x+y )	e^{-i\alpha D} \\& =
	e^{\frac{3}{2}\alpha}\phi_1(e^{\alpha}\Lambda_R x+e^{\alpha}y )
\\&=	 
e^{\frac{3}{2}\alpha}U(e^{\alpha}(x_0+   y_0),\mathbb{I}) \tilde{a}(e^{\alpha}(\mathbf{Rx+ y}))U(-e^{\alpha}(x_0+   y_0),\mathbb{I}) \\&=e^{\frac{3}{2}\alpha}
c_{y^{''} } \left(
\frac{	 1}{\vert e^{\alpha}(\mathbf{Rx+ y}) \vert^{2}-(e^{\alpha}(x_0+   y_0))^2 } \right)^2\ast\tilde{a} (e^{\alpha}(\mathbf{Rx+ y})),
\end{align*}
where $c_{y^{''} }:=\frac{i e^{\alpha}(x_0+   y_0)}{ \pi^2}$ and 
	where for the calculation we used the transformation property given in Theorem \ref{t3}, the dilatation action Equation (\ref{t2e2})
  and the explicit time translation given in Equation (\ref{eq51}).
\end{proof}
 
\section{Conclusion and Outlook}
In this work we supported the point of view, that the NWP operator and its respective states have a physical meaning.  In particular, we defined  the second-quantized version  of the NWP-operator for a massless scalar field and wrote the operator in its respective eigenbase. The meaning of the second-quantized operator is given by its respective eigenfunctions. They are used to calculate the probability amplitude of finding particles at certain spatial positions for a fixed time. Furthermore, by transforming the whole conformal group into the NWP-eigenbase we showed that this space can be treated and worked on at equal footing as the momentum space. 
\\\\
Although   eigenstates of this operator are non-covariant w.r.t. the Poincar\'e group we showed indirectly  that combinations of these states are covariant. This is proven by the fact that covariant operators remain covariant,  independently  of their respective representation. Hence, even though we work with non-covariant eigenstates, symmetries of the theory remain untouched. This is, in our opinion, another strong argument why the coordinate space, i.e. the NWP-space,  is equally entitled  from a physical point of view as the momentum space. 
\\\\
Moreover, we showed that the operator generating the NWP-space, i.e. $\phi_1(x)$ transforms in a covariant manner under the subgroup of space-time translations, rotations and dilatations of the conformal group, namely  
	\begin{align*}&
	e^{i\alpha D}	U(y,\Lambda_R) \phi_1(x)U(y,\Lambda_R)^{-1}	e^{-i\alpha D}
	= e^{\frac{3}{2} \alpha}\phi_1(	e^{\alpha}(\Lambda_R x+ y))
	.
	\end{align*}
Hence, given two observers that are related  by space-time translations, rotations or/and dilatations, they both measure the same  probability amplitude, by taking the respective covariant transformation into account. \\\\	
Furthermore, we were able to supply an additional argument to point out the physical relevance of the NWP-operator. Namely, it is possible, even in a second-quantized context, to express the generators 
of the  conformal group by using the NWP-operator and the relativistic momentum. In our opinion this supplies a strong argument for the physical sense of the NWP-operator, since the conformal group can be written in terms of products of the coordinate and momentum operators.  \\\\
	Since we formulated the whole framework of a free massless scalar field, its respective transformations and symmetries into the NWP-space, the next step in progress is to understand localization from a non-commutative point of view.  In particular, by imposing a space-space non-commutativity the subject of interest is to be able to calculate the probability of finding a particle at a certain mean value of a spatial position. In our opinion, the ground work for such investigations was presented in this paper. 

	\section*{Aknowledgments}
	The author is indebted to  Prof. K. Sibold for initiating   deep conceptual questions of this work.  Furthermore, we would like to thank Prof. K. Sibold  and S. Pottel for many fruitful discussions during different stages of this paper. With regards to the domains of the Lorentz-boosts am indebted to Prof. Norbert Dragon.  The   corrections of Dr. Z. Much are thankfully acknowledged.

\bibliographystyle{alpha}
\bibliography{allliterature1}

\end{document}